\documentclass[a4paper,12pt]{article}
\usepackage{graphicx}
\usepackage{amsfonts,amsmath,amsthm,amssymb}
\usepackage{multirow}
\usepackage{booktabs}
\usepackage{tikz}
\usepackage{fullpage}
\usepackage{mathtools} 

\newcommand{\C}{\mathcal{C}}

\newcommand{\F}{\mathcal{F}}

\newcommand{\Pa}{\mathcal{P}}

\newcommand{\R}{\mathbb{R}}

\newcommand{\rh}{r_{\mathcal{P}}(\mathcal{C})}

\newtheorem{mydef}{Definition}

\newtheorem{myth}{Theorem}
\newtheorem{lemma}{Lemma}

\newtheorem{cor}{Corollary}

\newtheorem{example}{Example}
\newtheorem*{ex1cont*}{Example 1 - cont}

\begin{document}

\title{Dynamically Consistent Objective and Subjective Rationality\thanks{%
\textit{Contacts}: \textit{lorenzo.bastianello@u-paris2.fr},  \textit{josehf@insper.edu.br}, \textit{anactros1@insper.edu.br}. 
We thank A. Chateauneuf, F. Maccheroni, N. Takeoka, P. Wakker and audiences at Insper, Lemma and at the TUS VI conference in Paris School of Economics. 
Santos gratefully acknowledges financial support from Fapesp (Grant No.
2017/09955-4 and 2018/00215-0). Faro is also grateful for financial support from the CNPq-Brazil (Grant No. 308183/2019-3).}}
\author{Lorenzo Bastianello$^a$, Jos\'{e} Heleno Faro$^b$,
Ana Santos$^b$ \\
$^{a}${\scriptsize Universit\'e Paris 2 Panthe\'eon-Assas, 4 rue
Blaise Desgoffe, 75006 Paris, France}\\
$^{b}${\scriptsize Insper, Rua Quat\'{a} 300, Vila Ol\'{\i}mpia
04546-042 S\~{a}o Paulo, Brazil}}
\date{April 2020
}
\maketitle

\begin{abstract}

 A group of experts, for instance climate scientists, is to choose among two policies $f$ and $g$. Consider the following decision rule. If all experts agree that the expected utility of $f$ is higher than the expected utility of $g$, the unanimity rule applies, and $f$ is chosen. Otherwise the precautionary principle is implemented and the policy yielding the highest minimal expected utility is chosen.  

This decision rule may lead to time inconsistencies when an intermediate period of partial resolution of uncertainty is added. We provide axioms that enlarge the initial group of experts with veto power, which leads to a set of probabilistic beliefs that is  ``rectangular'' in a minimal sense. This makes this decision rule dynamically consistent and provides, as a byproduct, a novel behavioral characterization of rectangularity. 

\medskip


\end{abstract}

\section{Introduction}\label{sec:intro}

Consider a board of Bayesian experts that needs to guide choices of a Decision Maker (DM henceforth) facing alternatives with uncertain outcomes. One can think for instance of a group of climate scientists that should advise the European Union about the best policy in order to fight CO$_2$ emissions. Another concrete and recent example may be the one of epidemiologists advising a prime minister about the best policy to fight the Coronavirus outbreak in 2020.
In both cases, it is reasonable to think that different experts assign different probability distributions to possible scenarios. For instance, in the beginning of March 2020, due to the lack of data and high uncertainty about the number of infected people, some epidemiologists may have assigned a high probability to a pandemic scenario with several hundred thousand of contaminated people by the Coronavirus in all Europe, while at the same time others may have thought that such a scenario was not as likely to happen.

Suppose that there are two policies $f$ and $g$ under consideration. For instance $f$ is lockdown the population and $g$ is herd immunity. If all experts think that $f$ is better than $g$,\footnote{When we say ``$f$ better than $g$'' we mean that $f$ has an higher expected utility than $g$.} then it may be reasonable for the DM to prefer policy $f$ rather than $g$. This decision rule is sometimes referred to as \textit{unanimity principle}. Note that all experts in the group have a veto power: it is sufficient that one expert ranks $g$ above $f$ to break unanimity. Suppose that this is indeed the case and  that  experts disagree. Then it is not clear which policy the DM should implement. In this situation, especially when uncertainty about different scenarios is high and there are scenarios that can lead to catastrophes, several authors suggest to adopt the \textit{precautionary principle}. While there is not an accepted and universal definition of the precautionary principle, see Gardiner \cite{Gardiner}, one can think of it as saying that a policy should be evaluated through the opinion of the most pessimistic expert.

Gilboa, Maccheroni, Marinacci,  and Schmeidler \cite{GMMS} (GMMS henceforth) offer an axiomatic analysis supporting the use of the precautionary principle in order to ``complete'' the unanimity rule when full agreement among experts does not hold.
However, if we introduce an intermediary period of partial resolution of uncertainty, i.e. if experts know today that they will have some information tomorrow, this decision rule may violate dynamic consistency. This means that decisions taken today may be regretted tomorrow once experts acquire information and are allowed to update their preferences.

This paper provides a refinement of GMMS. Loosely speaking, we show how to enlarge the group of experts in order to avoid decisions causing future regret. The main idea is that the board of experts should take into account today the structure of information describing all events that they might face tomorrow as the actual partial resolution of uncertainty.

From a decision theoretical point of view, one can identify each expert with the probability distribution she assigns to possible scenarios (under the assumption that experts share the same attitude toward risk, i.e. the same affine utility index). Hence the term expert and probability measure can be used as synonyms and a board of experts can be represented by a set of probability measures. Moreover decision rules can be thought of as  preference relations. In this paper the unanimity rule corresponds to a Bewley preference, see Bewley \cite{Bewley}, while the precautionary principle coincides with a Maxmin preference of Gilboa and Schmeidler \cite{GS89}.

We model information through an exogenous partition of the state space. The experts may disagree about the probability of any event in the partition but they must agree that all events in the partition have a positive probability. The DM knows today that tomorrow she will learn in which element of the partition lies the true state.  Our main result says that the unanimity rule should be modified in order to make the set of probabilities rectangular with respect to the information partition, in the sense of Sarin and Wakker \cite{SW} and Epstein and Schneider \cite{ES03}. This is equivalent to expanding the set of probabilities by considering all new probabilities formed by specific convex combinations of the Bayesian update (with respect to the elements of the partition) of experts' opinions.

 We provide two axioms, Coherence and Prudence, that characterize how the group of experts should be enlarged in order to obtain a dynamic consistent decision rule. These axioms are imposed on two Bewley preference relations, $\succsim^*$ and $\succsim^{**}$ (and on their generalized Bayesian updates $\succsim^*_E$ and $\succsim^{**}_E$), that correspond to the original and to the enlarged group of experts respectively. 
Coherence  is divided into two parts: Ex-Ante and Ex-Post Coherence. Ex-Ante Coherence says, loosely speaking, that the (new) unconditional preference $\succsim^{**}$ should agree with $\succsim^*$ on the likelihood of events in the information partition. Ex-Post Coherence states that the  (new) updated  preference $\succsim^{**}_E$ should agree with  $\succsim^*_E$. Therefore Coherence imposes some restrictions on how new experts should be chosen.
Prudence simply asserts that the novel unanimity rule will never reveal a ranking that is not supported by the original criteria. For the representation it implies that the new set of probability measures should contain the original one. It means that we may observe new experts added to the original group and all experts that were already there will still have veto power.
 
Finally, we require $\succsim^{**}$ to be the most incomplete preference such that the pair  $(\succsim^*~,~\succsim^{**})$ satisfies Coherence and Prudence. The idea is that we must add as many experts as possible until the corresponding precautionary completion becomes not subject to regret, maintaining the coherence and the prudence properties.  We call $\succsim^{**}$ the \textit{coherent precautionary restriction} of $\succsim^*$.
 On the one hand, extending the set of probabilities means that more opinions are taken into account. On the other hand, it implies that it is more likely that two experts disagree on the ranking of a pair of acts. Our point here is that the decrease of comparable acts (i.e. increasing incompleteness of preferences) is consistent with a DM whose behavior is driven by uncertainty aversion and regret aversion. Note that an alternative way to solve both the problems of incompleteness and time inconsistency is to pick only one expert out of the group. Since it seems that there is not an objective procedure to select the ``best expert", it seems reasonable for a DM to opt for a plurality of opinions.

Our main result states that given two Bewley preferences $\succsim^{*}$ and $\succsim^{**}$, the relation $\succsim^{**}$ is the coherent precautionary restriction of  $\succsim^{*}$ if, and only if, it is a Bewley preference represented by the same utility index on consequences and with a set of priors that is the rectangular hull of the set of probabilities characterizing the original preference $\succsim^{*}$. As a corollary, a simple application of GMMS allows us to obtain a dynamic consistent Maxmin preference as a completion of $\succsim^{**}$.\footnote{Completion means that if two acts cannot be compared by $\succsim^{**}$, then they should be compared using the derived Maxmin preference} Our results can be interpreted in the following way: if a set of expert is not rectangular and if one is willing to safely (i.e. without generating  dynamic inconsistent decisions) complete the unanimity rule by the precautionary principle, then this set should be enlarged through a ``rectangularization'' \textit{\`a la} Sarin and Wakker \cite{SW} and Epstein and Schneider \cite{ES03}. Rectangularization implies that in order to ensure dynamic consistency is enough to consider some convex combinations of the Bayesian updates of the original probabilities.

The rest of the paper is organized as follows. Section \ref{sec:framework} introduces the necessary notations. Section \ref{sec:GMMS} recalls the results of GMMS and provides an example of dynamic inconsistency. Section \ref{sec:main} contains our axioms and main results. Section \ref{sec:concludes} concludes.


\section{Framework and Notation}\label{sec:framework}

Consider a set $S$ of \emph{states of the world}, endowed with a $\sigma $-algebra $\Sigma $ of subsets called \emph{events}, and a non-empty set $X$ of \emph{consequences}. 
 We say that a function $f:S\rightarrow X$ is \textit{simple} if $f(S):=\{f(s):s\in S\}$ is a finite set. A simple function $f$ is $\Sigma $-measurable if $\{s\in
S:f(s)=x\}\in \Sigma $ for all $x\in X$. We denote by $\mathcal{F}$ the set of all simple and $\Sigma $-measurable functions.

We assume that set of consequences $X$ is a convex subset of a
vector space. For instance, this is the case if $X$ is the set of all simple lotteries on a set of outcomes $Z$. In fact, it is the classic setting of Anscombe and Aumann \cite{AA} as re-stated by Fishburn \cite{fish70}. Using the linear structure of $X$, we can define as usual for every $f,g\in \mathcal{F}$ and $\alpha \in \lbrack 0,1]$ the act $\alpha f+(1-\alpha)g:S\rightarrow X$ by 
\begin{equation*}
(\alpha f+(1-\alpha )g)(s)=\alpha f(s)+(1-\alpha )g(s).
\end{equation*}
Also, given two acts $f,g\in \mathcal{F}$ and an event $E\in \Sigma $, we denote by $fEg$ the act delivering the consequences $f\left( s\right) $ in $E $ and $g\left( s\right) $ in $E^{c}:=S\backslash E$ (the complement of $E$).

We denote by $B_{0}(\Sigma )$ the set of all simple real-valued $\Sigma$-measurable functions $a:S\rightarrow \mathbb{R}$. The norm in $B_{0}(\Sigma)$ is given by $\left\Vert a\right\Vert _{\infty }=\sup_{s\in S}\left\vert a(s)\right\vert $ (called the \textit{sup norm}) and $B\left( \Sigma \right)$ will denote the supnorm closure of $B_{0}\left( \Sigma \right) $. In
another way, $B_{0}(\Sigma )$ is the vector generated by the indicator functions of the elements of $\Sigma $, endowed with the supnorm (for more details, see Dunford and Schwarts \cite{DS}, Section 5 of Chapter IV). We denote by $ba(\Sigma )$ the Banach space of all finitely additive set functions on $\Sigma $ endowed with the total variation norm. It is
isometrically isomorphic to the norm dual of $B_{0}(\Sigma )$. Note also that the weak$^{\ast }$ topology $\sigma (ba,B_{0})$ of $ba(\Sigma )$ coincides with the eventwise convergence topology. Throughout the paper, we assume that any subset of $ba(\Sigma )$ is endowed with the topology inherited from the weak$^{\ast }$ topology.

Given a mapping $u:X\rightarrow \mathbb{R}$, function $u(f):S\rightarrow  \mathbb{R}$ is defined by $u(f)(s)=u(f(s)),\,$for\thinspace all\thinspace $s\in S$. We note that $u(f)\in B_{0}(\Sigma )$ whenever $f$ belongs to $\mathcal{F}$. Let $x$ be a consequence in $X$, abusing notation we define $x\in \mathcal{F}$ to be the constant act such that $x(s)=x$ for all $s\in S.$
Hence, we can identify $X$ with the set of constant acts in $\mathcal{F}$. We say that a function $u:X\rightarrow \mathbb{R}$ is \textit{affine} if  for every $f,g\in \mathcal{F}$ and $\alpha \in \lbrack 0,1]$, $u(\alpha f+(1-\alpha)g)=\alpha u(f)+(1-\alpha)u(g)$. Affine functions $u:X\rightarrow \mathbb{R}$ are called \textit{utility functions}.

We denote by $\Delta \left( S,\Sigma \right) :=\Delta $ the set of all (finitely additive) probability measures $P:\Sigma \rightarrow \lbrack 0,1]$. Given an act $f\in \mathcal{F}$, a utility index $u$, and a probability measure $P\in \Delta $, the expected utility of $f$ is denoted by $\int u\left( f\right) dP$. Consider an event $E\in \Sigma$ and a probability $P\in \Delta$ such that $P(E)>0$. The \textit{Bayesian update} of $P$ with respect to $E$ is $P^{E}(A)=\frac{P(A\cap E)}{P(E)}$. Let $\C\subseteq\Delta$ and $E\in \Sigma$ such that $P(E)>0$ for all $P\in \C$, then the set $\mathcal{C}^{E}$ denotes the set of \textit{prior-by-prior Bayesian updates} of $\C$ given $E$, i.e. $\mathcal{C}^{E}=\{P^{E}|P\in \mathcal{C}\}$. We also say that $\C$ is updated following the \textit{full Bayesian rule}.

A \textit{preference relation} $\succsim \subseteq \mathcal{F}\times \mathcal{F}$ is a binary relation that satisfies reflexivity, transitivity (\textit{preorder}), continuity and non-triviality.  \textit{Continuity} means that for all $f,g,h\in \F$ the sets $\{\lambda\in[0,1]| \lambda f +(1-\lambda)g\succsim (\precsim) h\}$ are closed in $[0,1]$. \textit{Non-triviality} means that $\succsim$ has a non-empty strict part. As usual, the strict and weak parts of $\succsim$ are denoted $\succ$ and $\sim$ respectively.

\section{Objective and Subjective Rationality and Dynamic Consistency}\label{sec:GMMS}

This section discusses  the interplay between the unanimity rule and the precautionary principle and the role of dynamic consistency. Section \ref{subsec:GMMS} recalls the axiomatic analysis of GMMS in which it is studied how to complete a Bewley preference with a Maxmin preference. Section  \ref{subsec:DC} introduces an intermediate period of partial resolution of uncertainty and provides an Ellsberg-type example in which dynamic inconsistencies may arise.

\subsection{Completion of a Bewley preference by a Maxmin preference}\label{subsec:GMMS}

In the context of social decisions, the unanimity principle postulates that society should prefer $f$ to $g$ if every individual prefers $f$ to $g$. Consider a group of individuals (applying unanimity) in which each member has Subjective Expected Utility (SEU) preferences with possibly different probability distributions and utility functions. Danan \emph{et al.} \cite{Dananetal} proposed normative (Pareto) principles in order to  aggregate individual preferences into a unanimity rule in which  individuals' utilities capturing tastes are combined \textit{\`a la} Harsanyi \cite{Harsanyi} into one (social) utility function on consequences. This means that it is as if society is represented by one DM with a Bewley \cite{Bewley} preference in which the set of probability distributions is characterized by the convex hull of the subjective priors of the members.\footnote{%
This can be viewed as an application of Theorem 2 of Danan \emph{et al.} \cite{Dananetal}. Their results are actually more general as each preference of the group members can be itself a Bewley preference.}

Formally, let $u:X\rightarrow \R$ be a utility function and  $\mathcal{C}\subseteq \Delta $  be a  nonempty, convex, and $\sigma (ba,B_{0})$-compact subset of $\Delta$. We say that $\succsim^*$ is a \textit{Bewley  preference} represented by  $(u,\mathcal{C})$ if for all $f,g\in \F$
\begin{equation}\label{eq:Bewley}
f\succsim^* g \Leftrightarrow \int u(f)dP\geq \int u(g)dP\,\, \forall P\in \C, 
\end{equation}
The criterion (\ref{eq:Bewley}), axiomatized by Bewley \cite{Bewley}, says that $f$ is preferred to $g$ with respect to the preference $\succsim^*$ if and only if the expected utility of $f$ is higher than the expected utility of $g$ according to every probability $P\in \C$. If every probability distribution in $\C$ represents the opinion of an expert then $f$ is better than $g$ if and only if every expert ranks $f$ above $g$. This justifies the name unanimity rule: experts should all agree.

In general this decision rule is incomplete, i.e. does not rank every pair $f,g\in \F$.  It may happen that there are $P_1,P_2\in \C$ such that $\int u(f)dP_1> \int u(g)dP_1$ and $\int u(g)dP_2>\int u(f)dP_2$.\footnote{Note that all probabilities, and hence all experts have the same importance. One can generalize this decision rule assigning different weights to different experts as in Faro \cite{Faro15}.} If two acts cannot be compared, but a decision must be taken, then several authors have suggested to use the precautionary principle, or Maxmin, see for instance Gardiner \cite{Gardiner}. This principle states that $f$ is better than $g$ if and only if the minimum expected utility of $f$ is higher that the minimum expected utility of $g$, where the minimum is considered over all probabilities in a set $\C$.

Put formally, if $u:X\rightarrow \R$ is a utility function and  $\mathcal{C}\subseteq \Delta $  a  nonempty, convex, and $\sigma (ba,B_{0})$-compact subset of $\Delta$, then $\succsim^\# $ is a \textit{Maxmin preference} represented by  $(u,\mathcal{C})$ if for all $f,g\in \F$
\begin{equation}\label{eq:MaxMin}
f\succsim^\# g \Leftrightarrow \min_{P\in \C}\int u(f)dP\geq \min_{P\in \C} \int u(g)dP.
\end{equation}
Maxmin has been introduced by  Wald \cite{Wald} in statistical decision theory and it has been axiomatized in our framework by Gilboa and Schmeidler \cite{GS89}. It is easy to see that the preference relation $\succsim^\#$ represented by  (\ref{eq:MaxMin}) is  complete, i.e. it allows to compare every $f$ and $g$.

GMMS  analyze the interplay between a Bewley preference  $\succsim^*$ and a complete  preference $\succsim^\#$  (note that  in our definition of preference relation in Section \ref{sec:framework} completeness is not required). In GMMS,  $\succsim^*$ represents the \textit{objective rationality} of the DM: the DM can convince others that $f$ is better than $g$ in an uncontroversial way. The preference $\succsim^\#$ represents the \textit{subjective rationality} of the DM: the DM cannot be convinced of being wrong choosing $f$ rather than $g$ (the DM does not feel embarrassed after her choice).  GMMS study under which conditions $\succsim^\#$ is a Maxmin preference that can be used to compare acts that are not comparable with respect to $\succsim^*$. They impose two axioms on the couple $(\succsim^*,\succsim^\#)$.

\medskip
\noindent  \textsc{Consistency} For all $f,g\in \F$, $f\succsim^* g $ implies $f\succsim^\# g$
\medskip

\noindent  \textsc{Default to Certainty} For all $f\in \F$ and $x\in X$, if not $f\succsim^* x$ then $x\succ^\# f$.
\medskip

Consistency says that the decision rule $\succsim^\#$ agrees with $\succsim^*$ whenever acts are comparable. In the spirit of GMMS, if it is objectively rational to prefer $f$ to $g$, then it is subjectively rational too. Consistency implies that the preference $\succsim^\#$ is a \textit{completion} of $\succsim^*$, meaning that $\succsim^\#$ is complete and $\succsim^*\subseteq \succsim^\#$.  Default to certainty says that if there is not a unanimous agreement about the comparison between act $f\in \F$ and $x\in X$, then $\succsim^\#$ should rank the constant $x$ above the uncertain act $f$.

\begin{myth}\label{th:GMMS}\textsc{[GMMS, Theorem 4]}
Let $\succsim ^*$ be a Bewley preference represented by $(u,\mathcal{C})$ and let $\succsim^\#$ be a  complete preference relation. Then:
\begin{enumerate}
\item[(i)] The pair $(\succsim^*,\succsim^\#)$ jointly satisfies Consistency and Default to Certainty;
\item[(ii)]  $\succsim^\#$ is a Maxmin preference represented by $(u,\mathcal{C})$.
\end{enumerate}
\end{myth}


\subsection{Dynamic (in)Consistency}\label{subsec:DC}

The framework of Section \ref{subsec:GMMS} is static. It does not take into consideration how a DM would react to new information that could be obtained over time. Let us add an intermediate period of partial resolution of uncertainty.

Let $\mathcal{P}=\{E_{1},\dots ,E_{n}\}$ denote a finite partition of measurable events of $S$, i.e., $E_{1},...,E_{n}\in \Sigma $, $S=\cup_{i=1}^n E_i$ and $E_i\cap E_j=\emptyset$ for $i\neq j$. $\mathcal{P}$ denotes the structure of information, i.e. the DM knows today that tomorrow she will learn that $s\in E_i$ for some $i=1,\dots, n$. 

Consider a DM with a (unconditional or ex-ante) preference relation $\succsim$ over $\F$. Given $E\in \mathcal{P}$, we call $\succsim_E$ the \textit{conditional (or ex-post) preference} given $E$. It is interpreted as the preference of the DM once she knows $s\in E$. 
 
The following axiom, called \textit{Dynamic Consistency} is well known in the literature and  relates the unconditional preference $\succsim$ to the conditional preference $\succsim_E$.

\medskip
\noindent  \textsc{Dynamic Consistency} For all $f,g\in \F$ and $E\in \mathcal{P}$, $f\succsim_E g \Leftrightarrow fEg\succsim g$
\medskip

The axiom says that $f$ is better than $g$ conditional on $E$ if and only if whenever one replaces act $g$ with act $f$ on $E$, the resulting act $fEg$ is better than $g$. We refer the reader to Ghirardato \cite{Ghiro02} for a detailed interpretation of this property.

A second standard axiom well known in the literature is \textit{Consequentialism}. It says that a DM should not be concerned about the consequences of an act in states that are known not to occur. We denote $f|_E$  the restriction of act $f$ on event $E$.

\medskip
\noindent  \textsc{Consequentialism} For all $f,g\in \F$ and $E\in \mathcal{P}$, if $f|_E=g|_E$, then $f\sim_E g$. 
\medskip

\noindent As shown in Faro and Lefort \cite{FlF}, if the couple of preferences $(\succsim, \succsim_E)$ satisfies Dynamic Consistency, then $\succsim_E$ satisfies Consequentialism.

Let $\succsim^*$ be a  Bewley preference relation represented by $(u,\C)$. Consider an even $E\in \Sigma$ such that $P(E)>0$ for all $P\in\C$. Then Ghirardato, Maccheroni and Marinacci \cite{GMM08} prove that Dynamic Consistency and Consequentialism are equivalent of $\succsim^*_E$ being represented by $(u,\C^E)$, where $\C^E$ is the set obtained from $\C$ by the full Bayesian updating rule, i.e. $\C^E=\{P^E|P\in \C\}$ and $P^E(\cdot)=\frac{P(\cdot\,\cap E)}{P(E)}$. The issue about the equivalence between dynamic consistency and the prior-by-prior updating rule in Bewley's model was also previously discussed by Bewley \cite{Bewley_Part2} and Epstein and Le Breton \cite {Epstein_LeBreton}. 

Note that the conditional preference $\succsim^*_E$ is in general incomplete. If one is willing to directly use Theorem \ref{th:GMMS} in order to complete it with a Maxmin preference $\succsim^\#_E$ represented by $(u,\C^E)$, then possible dynamic inconsistency may arise. Consider the following dynamic Ellsberg \cite{Ellsberg} example.\footnote{This example is inspired by Example 1 in Ghirardato \textit{et al.} \cite{GMM08}. Ghirardato \textit{et al.} \cite{GMM08} acknowledge  that they owe the example to Denis Bouyssou.}

\begin{example}\label{ex:DinC}
One of the main summary statistics in climate change science is the equilibrium climate sensitivity. It denotes the equilibrium increase in global mean temperatures that would occur if the concentration of atmospheric $\text{CO}_2$ were doubled. Uncertainty about estimates of this parameter is high. See Figure 1 in  Meinshausen \textit{el al.} \cite{Nature} in which are plotted estimated probability density functions for climate sensitivity from several published studies. Consider the following simple example. There are three possible future scenarios: low, medium and high climate sensitivity. A group of experts, with the same attitude toward risk (i.e. same $u(\cdot)$), knows that the probability $P(low)=\frac{1}{3}$, while they have no information about $P(medium)$ and $P(high)$. Therefore experts can be identified with the set
$$
\C=\left\{P=(p_1,p_2,p_3)\in \Delta| p_1=\frac{1}{3}, p_2\in\left[0,\frac{2}{3}\right]\right\}.
$$
This example is formally equivalent to an Ellsberg urn containing 90 balls, 30 of which are red, while the remaining 60 are either blue or green. Let us denote low, medium and high by $R$, $B$ and $G$ respectively.  Suppose that a government has to choose between the two policies in the table below. 
\begin{center}
\begin{tabular}{c|ccc}	  
	 &\textcolor{red}{Red} & \textcolor{blue}{Blue} & \textcolor{green}{Green} 	\\ \hline
  $f$ & 10 &  0 & 10  \\ 
   $g$ & 0 &  10 & 10  \\  \hline  
\end{tabular}
\end{center}
Note that $f$ and $g$ are not comparable with respect to a Bewley preference represented by $(u,\C)$. Consider for instance $P_1=\left(\frac{1}{3}, 0, \frac{2}{3}\right)$ and $P_2=\left(\frac{1}{3}, \frac{2}{3}, 0\right)$. Suppose w.l.g. $u(0)=0$, then $\int u(f)dP_1=u(10)>\frac{2}{3}u(10)=\int u(g)dP_1$ and $\int u(f)dP_2=\frac{1}{3} u(10)<\frac{2}{3}u(10)=\int u(g)dP_2$. If we use  a Maxmin preference in order to compare the two acts we get
$$
\min_{P\in \C}\int u(f)dP=\frac{1}{3}u(10)<\frac{2}{3}u(10)=\min_{P\in \C}\int u(g)dP
$$
i.e. $g\succ^\# f$.

Let us add an intermediary period of partial resolution of uncertainty. Information is modeled through the partition
$$
\Pa=\{G,RB\}.
$$
Suppose that the experts update the probabilities in $\C$ with the full Bayesian rule.
In our example one has $\C^G=\{(0,0,1)\}$, i.e. experts know that the true state is $G$,\footnote{For sake of simplicity we are considering full Bayesian updating even if there is $P\in\C$ with $P(G)=0$. It is easy to modify the example to avoid this possibility (for instance choosing $p_2\in [0,2/3)$ in $\C$).} 
and $\C^{RB}=\{(p_1,p_2,0)\in \Delta|p_1\in\left[\frac{1}{3},1\right] \}$. This results from $P^{RB}(R)=\frac{P(R)}{P(R)+P(B)}=\frac{\frac{1}{3}}{\frac{1}{3}+P(B)}$.

By Consequentialism we obtain $f\sim^*_G g$, but again $f$ and $g$ are not comparable with respect the Bewley preference $\succsim^*_{RB}$ represented by $(u,\C^{RB})$ (consider for instance $Q_1=(1,0,0)$ and $Q_2=\left(\frac{1}{3}, \frac{2}{3}, 0\right)$).  
Suppose that one wants to complete $\succsim^*_{RB}$ with a Maxmin preference $\succsim^{\#}_{RB}$ represented by $(u,\C^{RB})$. Then 
$$
\min_{P\in \C^{RB}}\int u(f)dP=\frac{1}{3}u(10)>0=\min_{P\in \C^{RB}}\int u(g)dP
$$
i.e. $f\succ^\#_{RB} g$. 
Note however that $g RB f~=~g\succ^\# f$ would imply $g\succ_{RB}^\# f$ if Dynamic Consistency were to hold.  In general using prior-by-prior updating and then applying Theorem~\ref{th:GMMS} violates Dynamic Consistency. 
\end{example}

An axiomatization for the full Bayesian updating in the model of objective and subjective rationality has been proposed by Faro and Lefort \cite{FlF} in a framework without the inclusion of a partition in the primitives. In their work, dynamic inconsistencies are allowed and interpreted as a product of what they call forced choices (decisions that must be made and based only on subjective grounds). In the perspective of this paper, Dynamic Consistency is viewed as a property of preferences fundamentally related to rationality. If a preference relation is not dynamic consistent then we may have $f\succ g$ but $g\succ_E f$ and $g\succ_{E^c} f$.\footnote{Note that in Example \ref{ex:DinC} 
we got $g\succ^\# f$ but $f\succ^\#_{RB} g$ and $f\sim^\#_{G} g$
 (by Consequentialism).}  This means that for dynamic inconsistent preferences, decisions taken today, i.e. choosing $f$ over $g$, may be regretted tomorrow, i.e. the conditional preference will rank $g$ above $f$ no matter if $s\in E$ or $s\in E^c$. We think therefore that it is reasonable to require Dynamic Consistency for  subjectively rational preferences (defined in the spirit of GMMS). 
 
It is well known that a Maxmin preference relation $\succsim^\#$  represented by $(u,\C)$ is not dynamically consistent in general. Epstein and Schneider \cite{ES03} prove that  Dynamic Consistency holds if and only if the set of priors $\C$ is rectangular. In Section \ref{sec:main} we show  how the unanimity rule $\succsim^*$ should be revised in order to achieve a rectangular set and therefore Dynamic Consistency for the derived Maxmin preference relation.

\section{Axioms and main result}\label{sec:main}

Let  $\succsim ^{\ast }$ be a Bewley preference represented by $(u,\mathcal{C})$. As shown in Example \ref{ex:DinC}, if the ex-post preference $\succsim ^{\ast }_E$  represented by $(u,\mathcal{C}^E)$  is completed by the corresponding Maxmin preference  represented by $(u,\mathcal{C}^E)$, dynamic inconsistencies may arise. We study here how $\succsim ^{\ast }$ should be transformed into a ``new'' Bewley preference $\succsim ^{\ast \ast }$ represented by $(\hat{u},\hat{\mathcal{C}})$ in order to avoid those inconsistencies. Without loss of generality, we denote $x_0\in X$ an outcome such that $u(x_0)=\hat{u}(x_0)=0$. 

Given a finite partition $\mathcal{P}=\{E_{1},\cdots
,E_{n}\}\subseteq \Sigma $, such that $P(E_i)>0$ for all $i \in \{1, \cdots, n\}$ and $P \in {\mathcal{C}}$, we impose the following axioms on the pairs of preferences $(\succsim ^{\ast },\succsim ^{\ast \ast })$ and $(\succsim ^{\ast }_E,\succsim ^{\ast \ast }_E)_{E\in \Pa}$.

\medskip
\noindent  \textsc{Coherence} :
\begin{itemize}
\item[(i)]  \textsc{Ex-Ante Coherence}.  For all $x,y\in X$, for all $ E \in \mathcal{P}$,  $y\succsim^* xE x_0\Rightarrow y\succsim^{**} xEx_0$ and $xEx_0 \succsim^* y\Rightarrow x E x_0\succsim^{**} y$.

\item[(ii)]  \textsc{Ex-Post Coherence}. For all $f,g\in \mathcal{F}$, for all $E\in \mathcal{P}$, $f\succsim_E^* g\Rightarrow f\succsim^{**}_E g$.
\end{itemize}

\medskip
\noindent  \textsc{Prudence.} For all $f,g\in \mathcal{F}$, $f\succsim^{**}g\Rightarrow f\succsim^* g$.
\medskip

Coherence imposes some restrictions on the new preference relation $\succsim^{**}$. Note that we can interpret an act $xE x_0$ as a bet over $E$. For instance $y\succsim^* xE x_0$ means that receiving $u(y)$ for sure is preferred to the bet $[u(x),P(E);0,1-P(E)]$ for all $P\in \C$.   Ex-Ante Coherence says that $\succsim^{**}$ should rank bets over events in $\Pa$ as $\succsim^{*}$.   Ex-Post Coherence simply says that the new conditional preference relation $\succsim^{**}_E$ should agree with  $\succsim^{*}_E$ for all $E\in \Pa$.  Coherence disciplines the relation between the sets   $\C$ and $\hat{\mathcal{C}}$ of ex-ante probabilities and the sets   $\C^E$ and $\hat{\mathcal{C}}^E$ of conditional probabilities. 

To give an interpretation of Prudence, let us interpret  $\C$ and $\hat{\mathcal{C}}$ as sets of (probabilistic) opinions of a group of experts.  Prudence essentially says that we want to ``add'' experts to $\C$. This will make  $\succsim^{**}$ more incomplete than $\succsim^{*}$, but will insure that more opinions are taken into account. In the terminology of Ghirardato, Maccheroni and Marinacci \cite{GMM04}, $\succsim^{**}$ reveals more ambiguity than $\succsim^{*}$. From their paper it is known that Prudence implies $u=\hat{u}$ and $\C\subseteq\hat{\mathcal{C}}$.

The following definition makes precise how the Bewley preference  $\succsim^{**}$ is built upon $\succsim^{*}$. 

\begin{mydef}\label{def:CohPrecRev}
Given a finite partition $\mathcal{P}=\{E_{1},\cdots ,E_{n}\}\subseteq \Sigma $, such that $P(E_i)>0$ for all $i \in \{1, \cdots, n\}$ and $P \in {\mathcal{C}}$, we say $\succsim ^{\ast \ast }$ is the \emph{coherent precautionary restriction} (w.r.t. $\mathcal{P}$) of $\succsim ^{\ast }$ if $\succsim ^{\ast \ast }$ is the most incomplete Bewley preference such that the pair $(\succsim ^{\ast },\succsim ^{\ast \ast })$ satisfies Coherence and Prudence.
\end{mydef}

Definition \ref{def:CohPrecRev} says that not only we  want the pair $(\succsim ^{\ast },\succsim ^{\ast \ast })$ to satisfy Coherence and Prudence, but also we  require $\succsim ^{\ast \ast }$ to be the maximal restriction of $\succsim ^{\ast }$ that obeys such properties (recall that $\succsim_1$ is a \textit{restriction} of  $\succsim_2$ if $\succsim_1\subseteq \succsim_2$. ). 
Loosely speaking, if we do not impose this condition, the relation $\succsim ^{\ast \ast }$ will be subject to the same issues of dynamic inconsistencies as the original preference $\succsim ^{\ast }$ once we take the corresponding completion with a Maxmin representation.

We can now state our main result.

\begin{myth}
\label{th:main_th}  Given a finite partition $\mathcal{P}%
=\{E_{1},\cdots ,E_{n}\}\subseteq \Sigma $, the following assertions are
equivalent:

\begin{enumerate}
\item[(i)] The preference $\succsim ^{\ast \ast }$ is the coherent precautionary restriction of $\succsim ^{\ast }$;

\item[(ii)] For all $f,g\in \mathcal{F}$, 
\begin{multline}\label{eq:BewleyDC}
f\succsim ^{\ast \ast }g \Leftrightarrow \\
 \sum_{i=1}^{n}P_{0}(E_{i})\int u(f)dP_{i}^{E_{i}}\geq \sum_{i=1}^{n}P_{0}(E_{i})\int u(g)dP_{i}^{E_{i}} \text{, }
\forall P_{0},P_{1},\dots ,P_{n}\in \mathcal{C} 
\end{multline}
\end{enumerate}
\end{myth}

Expression (\ref{eq:BewleyDC}) in Theorem \ref{th:main_th} tells how acts should be evaluated by the coherent precautionary restriction. First one should fix $n+1$ probabilities  in $P_{0},P_{1},\dots ,P_{n}\in \mathcal{C}$, i.e. $n+1$ experts should be chosen. Each probability $P_i$, $i=1,\dots,n$ should be assigned to a set in the partition $\Pa$ (for simplicity we denote $P_i$ the probability assigned to set $E_i$). Then for act $f$ the quantity $\sum_{i=1}^{n}P_{0}(E_{i})\int u(f)dP_{i}^{E_{i}}$ should be computed. Expression $\int u(f)dP_{i}^{E_{i}}$ is the expected utility of $f$ calculated through the  Bayesian update  with respect to $E_i$ of the corresponding probability $P_i$. Then expected utilities are aggregated through a convex combinations in which weights are given by $P_0(E_i)$. To summarize: 
\begin{equation}\label{eq:EUrect}
\underbrace{\sum_{i=1}^{n}P_{0}(E_{i})}_{\mathclap{\text{cvx comb.}}}\underbrace{
\vphantom{ \sum_{i=1}^{n}}
\int u(f)dP_{i}^{E_{i}}}_{\substack{\text{EU of } f \text{ w.r.t. } \\ \mathclap{\text{ update }P_i^{E_{i}}}}}
\end{equation}
The amount in formula (\ref{eq:EUrect}) should be calculated for all possible $n+1$ choices of probabilities in $\C$, for both acts $f$ and $g$. If for all choices of probabilities the value obtained for $f$ is higher than the one obtained for $g$, then $f\succsim ^{\ast \ast }g$. This result can be viewed as a prescriptive way about how to aggregate opinions. A Bewley preference $\succsim^*$ represented by $(u,\C)$ should be revised in the following way: first, compute the Bayesian update for all the events in the partition and all members $P\in \C$; second, compute the expected value under these conditional probabilities; and third, take convex combinations using as weights the opinions of the members on likelihood of events in $\Pa$. Therefore  (\ref{eq:BewleyDC}) is a  decision criterion in which new experts acquire veto power. These new experts are ``constructed'' from the old ones precisely as we just described.

Given a finite partition $\mathcal{P}=\{E_{1},\cdots ,E_{n}\}\subseteq \Sigma $, it is well known  that a probability $Q$ can be written as $Q=\sum_{i=1}^{n}P_{0}(E_{i})P_i^{E_{i}}$ for some $P_{0},P_{1},\dots ,P_{n}\in \mathcal{C}$  if and only if $Q$ is in the rectangular hull of $\C$. This notion is formalized in the Definition \ref{def:recthull}, based on the previous contributions of Sarin and Wakker \cite{SW}, Epstein and
Schneider \cite{ES03} and Ghirardato \textit{et al.} \cite{GMM08}. 

\begin{mydef}\label{def:recthull}
Given a partition $\mathcal{P}=\{E_{1},\cdots ,E_{n}\}\subseteq \Sigma $, the rectangular hull of a set of priors $\mathcal{C}\subseteq \Delta $ is
given by 
\begin{equation*}
r_{\mathcal{P}}(\mathcal{C}):=\left\{ \sum_{i=1}^{n}P_0(E_{i})\cdot
P_{i}^{E_{i}}\text{: }P_0,P_{1},\cdots ,P_{n}\in \mathcal{C}\right\} 
\text{.}
\end{equation*}
We say that a set $\mathcal{C}\subseteq \Delta $ is rectangular
(w.r.t. $\mathcal{P}$) when $\mathcal{C=}r_{\mathcal{P}}(\mathcal{C})$.
\end{mydef}

The rectangular hull of the set $\mathcal{C}$ for a partition $\mathcal{P}$ is obtained by considering all convex combinations   of the conditional probabilities (conditioned using Bayesian updating on events in $\Pa$) with weights given by the unconditional probabilities. The link with (\ref{eq:EUrect}) should be evident. Given Definition \ref{def:recthull} and Theorem \ref{th:main_th} the following corollary is immediate.

\begin{cor}\label{cor:rh}
Item $(ii)$ of Theorem \ref{th:main_th} is equivalent to 
\begin{enumerate}
\item[(iii)]  For all $f,g\in \mathcal{F}$, 
\end{enumerate}
\begin{equation}
f\succsim ^{\ast \ast }g \Leftrightarrow  \int u(f)dQ\geq \int u(g)dQ\text{, }\forall Q\in r_{\mathcal{P}}(\mathcal{C})
\end{equation}
\end{cor}

Obviously $\C\subseteq  r_{\mathcal{P}}(\mathcal{C})$, which reflects the fact that the coherent precautionary restriction  $\succsim ^{\ast \ast }$ is a sub-relation of  $\succsim ^{\ast }$. This comes from Prudence and from Definition \ref{def:CohPrecRev}. If we want to complete $\succsim ^{\ast \ast }$ as we did in Section \ref{subsec:GMMS}, an application of Theorem \ref{th:GMMS} yields the following result.

\begin{cor}\label{cor:DC_GMMS}
Assume that $\succsim ^{\ast }$ is a Bewley preference represented by $(u,\mathcal{C})$. Given a finite partition $%
\mathcal{P}=\{E_{1},\cdots ,E_{n}\}\subseteq \Sigma $, let $\succsim ^{\ast \ast}$ be the coherent precautionary restriction of$\ \succsim ^{\ast }$ and assume that $\succsim ^{\#\#}$ is a complete preference relation. The following are
equivalent:
\begin{enumerate}
\item[(i)] The pair $(\succsim^{**},\succsim^{\#\#})$ jointly satisfies Consistency and Default to Certainty;
\item[(ii)]  $\succsim ^{\#\#}$ is a Maxmin preference represented by $(u,r_{\mathcal{P}}(\mathcal{C}))$.
\end{enumerate}
Moreover Dynamic Consistency is satisfied  and for all $E\in \mathcal{P}$, $\succsim _{E}^{\#\#}$ is a Maxmin preference represented by $(u,r_{\mathcal{P}}(\mathcal{C})^{E})$.
\end{cor}

The last sentence of Corollary \ref{cor:DC_GMMS} says that  Dynamic Consistency is satisfied. This result is derived from Epstein and Schneider \cite{ES03}, see also Amarante and Siniscalchi \cite{AS19}.\footnote{For a recent application of Dynamic Consistency and rectangularity in a Maxmin model with imprecise probabilistic information, see Riedl, Tallon and Vergopoulos \cite{RTV}, in which the authors extend to dynamic settings the model of Gajdos \textit{et al.} \cite{Gaj}.} As we argued in Section \ref{subsec:DC}, Dynamic Consistency is an important property that a subjectively rational preference should satisfy.  Note that Corollary \ref{cor:DC_GMMS} together with Theorem \ref{th:main_th} give a new behavioral characterization of rectangularity and Dynamic Consistency. We conclude this section with two examples. The first one solves the dynamic inconsistency of Example \ref{ex:DinC} applying our results. The second one shows how the introduction of an information partition may increase the perceived ambiguity.

\begin{ex1cont*}\label{ex:DinC2}
In the Ellsberg example we get that $\bar{P}\in r_{\Pa}(\C)$ if and only if there are $P_0,P_1,P_2\in \C$ such that
\begin{align*}
\bar{P}(R) &= P_0(G)P_1^G(R)+(1-P_0(G))P_1^{RB}(R)=0+(1-P_0(G))P_1^{RB}(R)\\
\bar{P}(B) &= P_0(G)P_1^G(B)+(1-P_0(G))P_1^{RB}(B)=0+(1-P_0(G))P_1^{RB}(B)\\
\bar{P}(G) &= P_0(G)P_1^G(G)+(1-P_0(G))P_1^{RB}(G)=P_0(G)\\
\end{align*}
Therefore the coherent precautionary restriction $\succsim^{**}$ is a Bewley preference represented by $(u,r_{\Pa}(\C))$ with
$$
r_{\Pa}(\C)=\left\{(1-t)\begin{pmatrix}0\\ 0\\1\end{pmatrix}+t\begin{pmatrix}p\\ 1-p\\0\end{pmatrix} | t\in\left[\frac{1}{3},1\right], p\in\left[\frac{1}{3},1\right] \right\}.
$$
Acts $f$ and $g$ are of course not comparable with respect to the Bewley preference $\succsim^{**}$. Computing the Maxmin formula obtained by Corollary \ref{cor:DC_GMMS} one gets
$$
I(g)=\min_{t\in\left[\frac{1}{3},1\right],\, p\in\left[\frac{1}{3},1\right]} t(1-p)u(10)+(1-t)u(10)=0
$$

$$
I(f)=\min_{t\in\left[\frac{1}{3},1\right],\, p\in\left[\frac{1}{3},1\right]} tpu(10)+(1-t)u(10) > 0=I(g)
$$
Hence $f\succ^{\#\#} g$ and no Dynamic Consistency problem will  arise once the set $r_{\Pa}(\C)$ is updated with the prior-by-prior Bayes rule, since $r_{\Pa}(\C)^E=\C^E$ for all $E\in \Pa$.
\end{ex1cont*}

\begin{example}\label{ex:DC_ambiguity}
Consider an Ellsberg setting as in Example \ref{ex:DinC} and the following pair of acts
\begin{center}
\begin{tabular}{c|ccc}	  
	 &\textcolor{red}{Red} & \textcolor{blue}{Blue} & \textcolor{green}{Green} 	\\ \hline
  $f'$ & 10 &  0 & 0  \\ 
   $g$ & 0 &  10 & 10  \\  \hline  
\end{tabular}
\end{center}
If $\succsim^*$ is a Bewley preference represented by $(u,\C)$ (where $\C$ is as in  Example \ref{ex:DinC}) it is easy to show that $g\succ^*f'$.  Moreover, one can notice that acts $f'$ and $g$ are not ambiguous with respect to the probability set $\C$, meaning that $\int u(f) dP=\int u(f) dQ$ for all $P,Q\in \C$. This happens because the events $R$ and $BG$ are not ambiguous.
Clearly, if  $\succsim^\#$ denotes a Maxmin preference represented by the same pair $(u,\C)$, one obtains $g\succ^\# f'$.
However, if the information partition is $\Pa=\{G,RB\}$ there is an increase in ambiguity since $\Pa$ is not aligned with the ex-ante structure of information given by $\C$. Note actually that $f'$ and $g$ are not comparable with respect to $\succsim^*_{RB}$ and moreover one has $f'\succ^\#_{RB} g$, where $\succsim^*_{RB}$ and  $\succsim^\#_{RB}$ are the ex-post Bewley and Maxmin preferences represented both by $(u,\C^{RB})$, and $\C^{RB}$ is as in  Example~\ref{ex:DinC}. This violates Dynamic Consistency. Ghirardato \textit{et al.} \cite{GMM08} note
\begin{quotation}
``Yet, it seems to us that a decision maker with the ambiguity averse preferences [$f\succ g \succ f'$] \textit{might} still prefer to bet on a red ball being extracted, finding that event less ambiguous than the extraction of a [blue] ball, and that constraining him to choose otherwise is imposing a strong constraint on the dynamics of his ambiguity attitude. \\
In view of this example, we claim that dynamic consistency is a compelling property only for comparisons of acts that are not affected by the possible presence of ambiguity. In other words, we think that rankings of acts unaffected by ambiguity should be dynamically consistent."
\end{quotation}
Our approach is complementary to the one of Ghirardato \textit{et al.} \cite{GMM08}. In order to obtain Dynamic Consistency of the (uncertainty averse) subjectively rational preference, Corollary \ref{cor:rh} says that the coherent precautionary restriction $\succsim^{**}$ of $\succsim^{*}$ should be computed. Acts $f'$ and $g$ are not comparable with respect to $\succsim^{**}$, as for instance $g$ is ranked higher than $f'$ for the probability $(\frac{1}{9},\frac{2}{9},\frac{6}{9})\in r_{\Pa}(\C)$ (take $p=\frac{1}{3}$ and $t=\frac{1}{3}$) and $f'$ is ranked higher than $g$ for the probability $(1,0,0)\in r_{\Pa}(\C)$ (take $p=1$ and $t=1$). When the completion $\succsim^{\#\#}$ (a Maxmin preference represented by $(u,r_{\Pa}(\C)))$ is used, we obtain $f'\succ^{\#\#}g$. The ex-ante ranking is therefore reversed in a way that reflects the increase in ambiguity. This implies no violations of Dynamic Consistency. 
\end{example}

\section{Conclusion}\label{sec:concludes}

The unanimity rule (or equivalently Bewley preference relation) says that a policy (or equivalently act)  $f$ is preferred to $g$ if and only if every expert (or equivalently probability measure) assigns higher expected utility to $f$ rather than $g$. If two experts disagree, this rule is unable to tell which policy is better. When a decision must be taken, several authors suggest to compare policies through the precautionary principle: the policy with the highest minimal expected utility should be chosen.

We show that this rule may generate possible dynamic inconsistencies when an intermediary period of partial resolution of uncertainty is added. This  implies  that choices made today are regretted tomorrow no matter the additional information learned. In order to avoid this problem, we provide axioms that modify the original group of experts. We derive a new unanimity rule called coherent precautionary restriction. New opinions are formed by taking convex combinations of experts' updated beliefs. This makes the completion of the new unanimity rule dynamically consistent.

\section{Appendix\protect\smallskip}

Throughout this Appendix, $P=\{E_1\dots,E_n\}$ denotes a fixed partition of $S$. Preferences $\succsim^*$ and $\succsim^{**}$ are Bewley preferences represented by $(u,\C)$ and $(\hat{u},\hat{\C})$ respectively. 

\begin{lemma}
\label{lemma:C_hat_properties} The pair $(\succsim^*,\succsim^{**})$ satisfies Coherence and Prudence if and only if

\begin{itemize}
\item[(i)] $\hat{u}=u$;

\item[(ii)] $\hat{\mathcal{C}}\supseteq \mathcal{C}$;

\item[(iii)] For all $E\in\mathcal{P}$, $\mathcal{C}^E=\hat{\mathcal{C}}^E$;

\item[(iv)] For all $Q\in \hat{\mathcal{C}}$ there exists $P\in \mathcal{C}$
s.t. $P(E)=Q(E)$, for all $E\in \mathcal{P}$.
\end{itemize}
\end{lemma}

\begin{proof} We focus on the proof of sufficiency, as necessity can be easily proved.\\
\textit{We prove (i), (ii) and (iii)}. Using Ghirardato \textit{el al.} \cite{GMM04}, Proposition 6, we have that Prudence implies $\hat{u}=u$ and $\hat{\C}\supseteq\C$ and Ex-Post Coherence implies $\C^E\supseteq\hat{\C}^E$.  Moreover $\hat{\C}\supseteq\C$ and $\C^E\supseteq\hat{\C}^E$ imply $\C^E=\hat{\C}^E$. 

\medskip

\textit{We prove (iv)}.  Suppose by contradiction that there exist $Q\in \hat{\C}$ and $E\in \Pa$ such that $Q(E)\neq P(E)$ for all $P\in \C$. Then either $Q(E)> P(E)$ for all $P\in \C$, or $Q(E)< P(E)$ for all $P\in \C$. In fact, if for some $P_1, P_2$ we have $P_1(E)>Q(E)> P_2(E)$, then there exists $\alpha\in(0,1)$ such that $Q(E)=\alpha P_1(E)+(1-\alpha)P_2(E)$ and, by convexity of $\C$, $\alpha P_1+(1-\alpha)P_2\in \C$.\\
Suppose therefore $Q(E)> P(E)$ for all $P\in \C$ and let $\bar{P}(E)=\max_{P\in\C} P(E)$. Choose $x,y\in X$ such that $x\succ^* x_0$, $y\succ^* x_0$ and $\frac{u(x)}{u(y)}\in (\bar{P}(E),Q(E))$. Then  $\frac{u(x)}{u(y)}>\bar{P}(E)$ and this implies $u(x)>P(E)u(y)$ for all $P\in \C$, i.e. $x\succ^*yEx_0$. On the other hand $\frac{u(x)}{u(y)}<Q(E)$ implies $x\not\succsim^{**}yEx_0$, contradicting Ex-Ante Coherence.\\
Suppose now $Q(E)< P(E)$ for all $P\in \C$. Defining $\underline{P}(E)=\min_{P\in\C}P(E)$ and reasoning as before, one gets a contradiction with the second part of Ex-Ante Coherence.
\end{proof}

\begin{lemma}
\label{lemma:rh_max} $r_{\mathcal{P}}(\mathcal{C})$ is the maximal set such
that

\begin{itemize}
\item[(i)] $r_{\mathcal{P}}(\mathcal{C})^E=\mathcal{C}^E$ for all $E\in 
\mathcal{P}$;

\item[(ii)] $\forall Q\in r_{\mathcal{P}}(\mathcal{C})$, $\exists P\in 
\mathcal{C}$ such that $P(E)=Q(E)$, $\forall E\in \mathcal{P}$.
\end{itemize}
\end{lemma}

\begin{proof}
We first prove that $\rh$ satisfies conditions $(i)$ and $(ii)$ and then we show that it is the maximal set satisfying these conditions.

\medskip

\textit{$\rh$ satisfies condition $(i)$}. Fix $E_{j}\in \Pa$. We have  $P^{E_{j}}\in \rh^{E_{j}}$ if and only if there exists $P\in\rh$ such that $P^{E_{j}}(A)=\frac{P(A\cap E_j)}{P(E_j)}$ for all $A\in \Sigma$. The last assertion holds if and only if there exists $P_0,P_1,\dots,P_n\in \C$ such that 
\begin{equation}
\label{eq:rh_max}
P^{E_{j}}(A)=\frac{P(A\cap E_j)}{P(E_j)}=\frac{\sum_{i=1}^n P_0(E_i)P_i^{E_i}(A\cap E_j)}{\sum_{i=1}^n P_0(E_i)P_i^{E_i}(E_j)}=P_j^{E_j}(A\cap E_j)=P_j^{E_j}(A) 
\end{equation}
for all $A\in \Sigma$. This implies that $P^{E_{j}}\in \C^{E_j}$. \\
On the other hand if $P_j^{E_{j}}\in \C^{E_j}$ then, choosing $n+1$ probabilities $P_0,P_1,\dots,P_n\in \C$,  (\ref{eq:rh_max}) shows that  $P_j^{E_{j}}\in\rh^{E_{j}}$.

\medskip

\textit{$\rh$ satisfies condition $(ii)$}. Let $Q\in \rh$. Then there are probabilities $P_0,P_1,\dots,P_n\in \C$ such that for all $E_j\in \Pa$
$$
Q(E_j)=\sum_{i=1}^n P_0(E_i)P_i^{E_i}(E_j)=P_0(E_j)P_j^{E_j}(E_j)=P_0(E_j).
$$
Hence $P_0$ satisfies condition $(ii)$.

\medskip

\textit{$\rh$  is the maximal set satisfying conditions $(i)$ and $(ii)$}. Let $Q$ be a probability over the measurable space $(S,\Sigma)$ such that $Q^E\in \C^E$ for all $E\in \Pa$ and such that there exists $P\in \C$ such that $P(E)=Q(E)$  for all $E\in \Pa$.\\
For all $A\in \Sigma$, by the law of total probability we have $Q(A)=\sum_{i=1}^nQ(E_i)Q^{E_i}(A)$. Since $Q^{E_i}\in \C^{E_i}$ by $(i)$ there is $P_i\in \C$ such that $Q^{E_i}=P_i^{E_i}$. Moreover by condition $(ii)$ there is $P_0\in \C$ such that $Q(E_i)=P_0(E_i)$. This implies that $Q(A)=\sum_{i=1}^nP_0(E_i)P_i^{E_i}(A)$ and hence $Q\in \rh$.

\end{proof}

\medskip

\begin{proof}[\textbf{Proof of Theorem \ref{th:main_th}}]

 $\Rightarrow$ Let $\succsim^{**} $ be a Bewley preference represented by $(\hat{u},\hat{\mathcal{C}})$. Since $(\succsim^*,\succsim^{**})$ satisfies Coherence and Prudence Lemma \ref{lemma:C_hat_properties} implies that $\hat{u}=u$ and $\hat{\C}$ satisfies properties $(ii)$, $(iii)$ and $(iv)$ of Lemma \ref{lemma:C_hat_properties}. By Lemma \ref{lemma:rh_max}, $\rh$ is the maximal set satisfying  properties $(iii)$ and $(iv)$ of Lemma \ref{lemma:C_hat_properties}. By hypothesis,  $\succsim^{**} $ is the most  incomplete Bewley preference such that the pair  $(\succsim^*,\succsim^{**})$ satisfy Coherence and Prudence, hence $\hat{\C}\supseteq\rh$ and therefore $\hat{\C}=\rh$. This implies the result using Corollary \ref{cor:rh}.

\medskip

 $\Leftarrow$ By Corollary \ref{cor:rh} the representation of  $\succsim^{**} $ implies that $\forall f,g\in\F$,
$$
f \succsim^{**} g \, \Leftrightarrow \, \int u(f)dP \geq \int u(g) dP, \,\, \forall P\in \rh.
$$
It is obvious that properties $(i)$ and $(ii)$ of Lemma \ref{lemma:C_hat_properties} are satisfied. By Lemma \ref{lemma:rh_max}, $\rh$ satisfies properties $(iii)$ and $(iv)$ of Lemma \ref{lemma:C_hat_properties}. Since moreover $\rh$ is the maximal set satisfying these two properties, $\succsim^{**}$ is the most incomplete Bewley preference such that the pair  $(\succsim^*,\succsim^{**})$ satisfies Coherence and Prudence.

\end{proof}

\end{document}